%% file: main.tex
\documentclass[letterpaper, 10 pt, conference]{ieeeconf}
\IEEEoverridecommandlockouts
\usepackage{max_paper}

\title{\LARGE \bf
From Bundles to Backstepping: Geometric Control Barrier Functions\\ for Safety-Critical Control on Manifolds}

\author{Massimiliano de Sa, Pio Ong, and Aaron D. Ames%
\thanks{This research is supported by AFOSR Grant No. 113535-19668, Hybrid Dynamics - Deconstruction and Aggregation (HyDDRA).}%
\thanks{The authors are with the Department of Control and Dynamical Systems and the Department of Mechanical and Civil Engineering, California Institute of Technology, Pasadena CA 91125, U.S.A.
        {\tt\small \{mdesa, pioong, ames\}@caltech.edu}}%
}

\begin{document}

\maketitle
\thispagestyle{empty}
\pagestyle{empty}
\renewcommand{\bf}{\mathbf{f}}

\begin{abstract}

Control barrier functions (CBFs) have a well-established theory in Euclidean spaces, yet still lack general formulations and constructive synthesis tools for systems evolving on manifolds common in robotics and aerospace applications. In this paper, we develop a general theory of geometric CBFs on bundles and, for control-affine systems, recover the standard optimization-based CBF controllers and their smooth analogues. Then, by generalizing kinetic energy-based CBF backstepping to Riemannian manifolds, we provide a constructive CBF synthesis technique for geometric mechanical systems, as well as easily verifiable conditions under which it succeeds. Further, this technique utilizes mechanical structure to avoid computations on higher-order tangent bundles. We demonstrate its application to an underactuated satellite on SO(3).

\end{abstract}

\input{intro}
\input{prelim}
\input{geom-cbf}
\input{cbf_synth}
\input{example}
\input{conclusion}

\bibstyle{ieeetr}
\bibliography{references}

\end{document}

%% file: intro.tex
\section{Introduction}
In the modern nonlinear control landscape, safety-critical control is of ever-increasing importance. Control barrier functions (CBFs) \cite{ames2016control, ames2019control} have become a popular framework for designing safety-critical controllers. In this framework, safety is posed as \textit{set invariance} and is enforced through affine constraints well-suited to optimization-based control. 

Control barrier functions are typically studied in Euclidean state spaces. However, a wide range of problem settings have an underlying \textit{non-Euclidean} state space, leading geometry to play a nontrivial role in the control design and analysis. In problem settings ranging from control design for quadrotors \cite{lee2010geometric} and spacecraft \cite{lee2012exponential, bullo1995control, weiss2014spacecraft} to general design and analysis techniques for mechanical systems on manifolds \cite{welde2024almost}, geometry plays an important role in shaping both the practical and theoretical aspects of control.

Despite the rich literature available in Euclidean spaces, CBFs remain largely unexplored in geometric settings. Initial work in this direction focused on synthesizing high-order reciprocal CBFs for a class of fully actuated mechanical systems on manifolds \cite{wu2015safety}, as well as on synthesizing CBFs for a class of configuration constraints for quadrotors \cite{wu2016safety-critical}. However, in \cite{wu2015safety} only the fully-actuated case is considered, and in both \cite{wu2015safety, wu2016safety-critical}, the regularity of the proposed optimization-based CBF controllers goes unverified. Moreover, the ``zeroing" CBFs of \cite{ames2016control} have since become the favored formulation over the reciprocal CBFs considered in \cite{wu2015safety, wu2016safety-critical}, and the high-order CBF synthesis techniques used by \cite{wu2015safety} have been shown to encounter difficulties for simple systems and constraints \cite{cohen2024safety-review}.

In the time since this initial work on geometric CBFs, there have been significant advances in the Euclidean CBF theory that overcome many of the aforementioned challenges. In particular, backstepping \cite{taylor2022safe, cohen2024safety-review} has been shown to yield valid CBFs for high-order systems where the standard high-order CBF techniques fail. It has also demonstrated success in synthesizing CBFs from configuration constraints for Euclidean mechanical systems \cite{cohen2024safety-review, singletary2021safety}. Further, the regularity of CBF-based controllers has been extensively studied, with smooth analogues of safety filters being presented in \cite{cohen2023characterizing}.
\begin{figure}[t]
  \centering
  \includegraphics[scale=0.67]{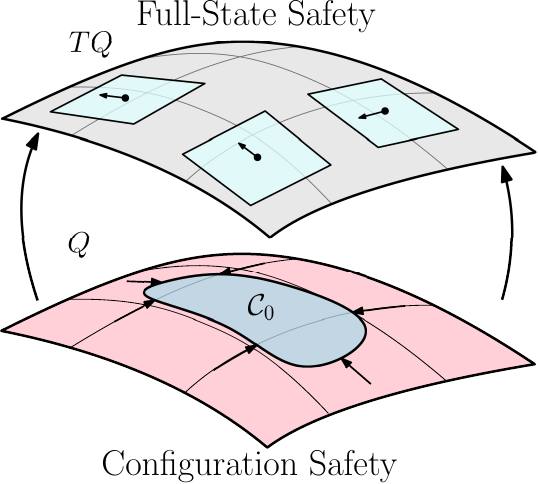}
  \caption{Using a Riemannian formulation of backstepping, we can lift a safe configuration set $\C_0 \subseteq Q$ on the configuration manifold of a geometric mechanical system to a control barrier function on its tangent bundle.}
  \label{fig:cbf-backst}
\end{figure}

In this work, we develop a rigorous framework for geometric CBFs. Unlike prior work, we develop CBFs in the general setting of bundles, enabling safety-critical control design for classes of geometric systems previously out of reach. Further, we present a constructive, global CBF synthesis technique for mechanical systems that avoids the challenges of high-order CBFs.
Our contributions are summarized as follows:
\begin{enumerate}
    \item Formulate CBFs for geometric control systems on bundles and characterize their safety properties.
    \item Develop computationally tractable closed-form solutions for CBF-QP controllers on vector bundles, and establish their regularity properties and smooth variants. 
    \item Devise a backstepping procedure for synthesizing CBFs from configuration constraints for a class of underactuated geometric mechanical systems, while avoiding computations on higher-order tangent bundles.
\end{enumerate}
The remainder of the paper is organized as follows. In Section II, we provide the background on geometry and Euclidean CBFs. In Section III, we develop geometric CBFs for control systems on bundles. In Section IV, we synthesize CBFs for geometric mechanical systems using backstepping. In Section V, the techniques are applied to satellite control.

%% file: prelim.tex
\section{Preliminaries}
\subsection{Differential Geometry}
To begin, we rapidly overview the background in differential geometry. Here, our aim is primarily to fix notation and provide informal definitions; we will detail the remaining prerequisite material as we proceed. For a thorough coverage and precise definitions, we refer to \cite{lee2013smooth, lee2018introduction}.

Without further qualification, we assume all objects to be smooth ($C^\infty$); we will occasionally substitute \textit{smooth} for \textit{locally Lipschitz},\footnote{For any map between smooth manifolds, \textit{Locally Lipschitz} is well-defined without specification of metric space structures \cite[p. 667]{kvalheim2021necessary}.} but shall be explicit when we do so. For a smooth manifold $\M$, we denote by $T_p \M$ its tangent space at $p \in \M$, $T_p^* \M$ its cotangent space at $p$, and by $T\M$ and $T^*\M$ its tangent and cotangent bundles, respectively. The action of a covector $\omega_p$ on a tangent vector $v_p$ is written $\braket{\omega_p; v_p}$. We write $\pi_{T\M}$ for the projection from $T\M$ to $\M$, taking $v_p \in T_p \M \mapsto p \in \M$. We write the global differential of a smooth map $F: \M_1 \to \M_2$ as $dF: T\M_1 \to T\M_2$, and its restriction to each tangent space as $dF_p$. The set of smooth functions from $\M \to \R$ is written $C^\infty(\M)$.

A \textit{bundle} is a tuple $(\pi, E, \M)$ of smooth manifolds $E$ and $\M$, and a smooth surjective submersion $\pi: E \to \M$. For each $p \in \M$, the set $\pi^{-1}(p) \subseteq E$, denoted $E_p$, is referred to as the \textit{fiber over $p$}. A \textit{section} of a bundle $(\pi, E, \M)$ is a map $\sigma: \M \to E$, for which $\pi \circ \sigma = \id_{\M}$, the identity on $\M$. The space of all sections of $(\pi, E, \M)$ is denoted $\Gamma(E)$.

Informally, a \textit{vector bundle of rank $k$} is a bundle $(\pi, E, \M)$ for which each fiber $E_p$ is a real, $k$-dimensional vector space. One can intuitively reason about a vector bundle as a family of vector spaces $E_p$ which vary smoothly with $p \in \M$. A \textit{subbundle}\footnote{Here, we consider subbundles in the sense of \cite{lee2013smooth}; the entire collection of subspaces of must form an embedded submanifold of $E$.} of a vector bundle is formed by smoothly assigning a subspace of $E_p$ to each $p\in \M$ to get a new vector bundle. A \textit{vector bundle map} between vector bundles $(\pi^1, E^1, \M)$ and $(\pi^2, E^2, \M)$ is a map $A: E^1 \to E^2$ for which $\pi^2 \circ A = \pi^1$ and for each $p \in \M$, the restriction $A_p \defeq A|_{E^1_p}: E^1_p \to E^2_p$ is a linear transformation.

Two key examples of vector bundles are the tangent and cotangent bundles, $T\M$ and $T^* \M$. A subbundle of $T\M$ is called a distribution, and a subbundle of $T^*\M$ a codistribution. Sections of $T\M$ are referred to as \textit{vector fields}, and the set of all vector fields on $\M$ is denoted $\X(\M)$. The value of a vector field $X$ at $p$ is written $X_p$. The maximal interval of existence of an integral curve of $X$ starting at $p$ is written $I(p)$, and its intersection with $\R_{\geq 0}$ as $I_{\geq 0}(p) \defeq I(p) \cap \R_{\geq 0}$. The \textit{flow} of $X$ at time $t$, starting from $p$, is written $\varphi_t(p)$. The Lie derivative of $h \in C^\infty(\M)$ along $X$, $\lie_X h$, is calculated $\lie_X h(p) = \tfrac{d}{dt}|_{t = 0}h(\varphi_t(p)) = dh_p X_p$ for each $p \in \M$.

A Riemannian manifold is a pair $(\M, g)$ of a manifold $\M$ and a \textit{Riemannian metric} $g$, a smooth assignment of an inner product to each $T_p\M$. We write $\braket{v_p, w_p}$ for the metric and $\norm{\cdot}$ for its induced norm on each $T_p \M$. Similarly, we define a metric on a vector bundle $(\pi, E, \M)$ as a smooth assignment of an inner product to each fiber $E_p$.

\subsection{Control Barrier Functions in Euclidean Space}
Consider the nonlinear system in Euclidean space with state $\bx \in \R^n$ and control input $\bu\in\R^m$:
\begin{equation}
    \label{sys:nonlinear}
    \dot \bx = \bPhi(\bx,\bu),
\end{equation}
where $\map{\bPhi}{\R^n\times \R^m}{\R^n}$ is the locally Lipschitz system dynamics. Given a locally Lipschitz controller $\map{\bk}{\R^n}{\R^m}$, the state feedback $\bu=\bk(\bx)$ ensures the existence and uniqueness of a flow $t \mapsto \varphi_t(\bx_0)$, defined for any initial condition $\bx_0 \in \R^n$ on a maximal interval of existence $I(\bx_0)$. Our focus is on designing controllers for \textit{safety} problems, wherein solutions are required to remain within a safety constraint $\Sc\subset\R^n$ for all positive times in $I(\bx_0)$.

One widely used approach to address safety problems is through control barrier functions. Specifically, we consider a continuously differentiable $\map{h}{\R^n}{\R}$ that defines the set:
\begin{align*}
    \Cc = \setdefb{\bx\in\R^n}{h(\bx)\geq 0}.
\end{align*}
To ensure safety, we require that $\Cc$ is a subset of $\Sc$ and that we can find a controller $\bk$ which renders $\Cc$ forward invariant. In this case, $\Cc$ serves as a certified safe operating region; any trajectory initialized in $\Cc$ and controlled by $\bk$ remains in $\Cc$ for all positive time. A sufficient condition for the existence of such controller is
that $h$ is a \textit{control barrier function}.
\begin{defn}\label{defn:euclidean_CBF}
    A continuously differentiable function $\map{h}{\R^n}{\R}$ is a \textit{control barrier function} (CBF) for system~\eqref{sys:ctrl-affine_euclidean} if there exists a function\footnote{A function $\alpha: \R \to \R$ belongs to $\K_\infty^e$ if it is continuous, strictly increasing, and satisfies $\alpha(0) = 0$, $\lim_{r \to \pm \infty}\alpha(r) = \pm \infty$.} $\alpha\in\K_\infty^e$  such that for each $\bx\in\Cc$:\begin{equation}\label{eq:CBF_condition_euclidean}
        \sup_{\bu\in\R^m} \frac{\partial h}{\partial \bx} \bPhi(\bx, \bu) > -\alpha(h(\bx)).
    \end{equation}
\end{defn}

Condition~\eqref{eq:CBF_condition_euclidean} is motivated by the fact that any locally Lipschitz controller $\map{\bk}{\R^n}{\R^m}$ satisfying the condition:
\begin{equation}\label{eq:BC_euclidean}
    \frac{\partial h}{\partial \bx} \bPhi(\bx, \bk(\bx))  \geq -\alpha(h(\bx)),
\end{equation}
for all $\bx$ in a neighborhood $\Dc$ of $\Cc$, will render set $\Cc$ forward invariant for the closed-loop system, and hence safe~\cite{ames2016control}. If $h$ is verified to be a CBF, there exists a $\bu$ satisfying~\eqref{eq:BC_euclidean} at each $\bx \in \C$. However, there is no general method for synthesizing these inputs into a continuous or locally Lipschitz function~$\bk$.

A constructive approach becomes available when the system dynamics $\bPhi$ is affine with respect to the control input:
\begin{equation}\label{sys:ctrl-affine_euclidean}
    \dot \bx = \bPhi(\bx,\bu)=\bf(\bx)+\bG(\bx)\bu,
\end{equation}
where $\map{\bf}{\R^n}{\R^n}$ and $\map{\bG}{\R^n}{\R^{n\times m}}$ are the drift and the actuation matrix, respectively. 
For control-affine systems, a locally Lipschitz controller can be constructed by solving a quadratic program (CBF-QP) that minimally modifies a desired controller $\map{\bk_{\rm{d}}}{\R^n}{\R^m}$ to enforce safety:
\begin{align}
\label{eq:CBF-QP_euclidean}
    \bk_\textup{QP}(\bx)  =\argmin_{\bu\in\R^m} \quad & \norm{\bu-\bk_{\rm{d}}(\bx)}_2^2                                                                  \\
             \textup{s.t.} \quad & \frac{\partial h}{\partial \bx} \bf(\bx) + \frac{\partial h}{\partial \bx} \bG (\bx) \bu \geq -\alpha(h(\bx)). \nonumber
\end{align}
When $h$ is a CBF, the CBF-QP is guaranteed to be feasible and locally Lipschitz in $\bx$ on an open set $\Dc\supset \Cc$. Thus, the CBF-QP ensures that trajectories starting in $\Cc$ remain in $\Cc$.

In Euclidean space, the theory of CBFs is well-established, but has not yet been fully translated into geometric settings. Some progress has been made in \cite{wu2015safety, wu2016safety-critical}, in which reciprocal barrier functions \cite[Defn. 1]{ames2016control} are studied for fully actuated geometric mechanical systems and for quadrotors evolving on $SE(3)$. 
In this work, we develop the theory of ``zeroing" CBFs (the class of CBF introduced above) in a general geometric setting. We start by formulating CBFs for control systems on bundles, and then discuss their synthesis for a class of underactuated geometric mechanical systems. 

%% file: geom-cbf.tex
\section{Geometric Control Barrier Functions}
Here, we generalize CBFs to the geometric setting. Our development parallels that of the Euclidean case, starting from barrier conditions that ensure forward invariance of safe sets and proceeding to the construction of safe controllers.

\subsection{Control Systems on Bundles}
We begin our exposition by precisely defining the control systems for which we will develop the geometric CBF theory. Following \cite{tabuada2005quotients}, we consider a class of nonlinear control systems defined on bundles.

\begin{defn}\label{defn:nonlinear-sys}
    A \textit{nonlinear control system} is a tuple $\Sigma =(\pi,\U,\M,\bPhi)$ consisting of a bundle $(\pi, \U, \M)$, (where $\U$ is termed the \textit{input space} and $\M$ the \textit{state space}) and the \textit{dynamics} $\bPhi: \U \to T\M$, a map for which the following diagram commutes ($\pi_{T\Mc} \circ \bPhi = \pi$):
    \[\begin{tikzcd}
        \U && T\Mc \\
        & \M
        \arrow["\bPhi", from=1-1, to=1-3]
        \arrow["\pi"', from=1-1, to=2-2]
        \arrow["{\pi_{T\Mc}}", from=1-3, to=2-2]
    \end{tikzcd}\]
    A \textit{controller} for $\Sigma$ is a section $\kappa \in \Gamma(\U)$ of $(\pi, \U, \M)$.
\end{defn}
Analogous to \eqref{sys:nonlinear}, the equations of motion for a nonlinear control system in the sense of Definition \ref{defn:nonlinear-sys} are written:
\begin{align}\label{eq:nonlinear-manifold}
    \dot p &= \bPhi(u_p), 
\end{align}
where $u_p$ belongs to $\U_p$, the fiber of the bundle $(\pi, \U, \M)$ over $p$. Due to the bundle structure, the input and state are packaged together in a single object, $u_p$. Since the input space $\U$ may not have a Cartesian product structure, we cannot separate state from input as we do in Euclidean space.
\begin{example}
    A control system $\dot \bx = \bPhi_{\text{Eucl.}}(\bx, \bu)$ on $\R^n \times \R^m$ is a nonlinear control system in the sense of Definition \ref{defn:nonlinear-sys}. Here, $\M = \R^n$, $\U = \R^n \times \R^m$, $\pi(\bx, \bu) = \bx$, and the $\bPhi$ is the map $(\bx, \bu) \mapsto (\bx, \bPhi_{\text{Eucl.}}(\bx, \bu)) \in T\R^n \approx \R^n \times \R^n$.
\end{example}

\begin{figure}
    \centering
    \includegraphics[width=0.78\linewidth]{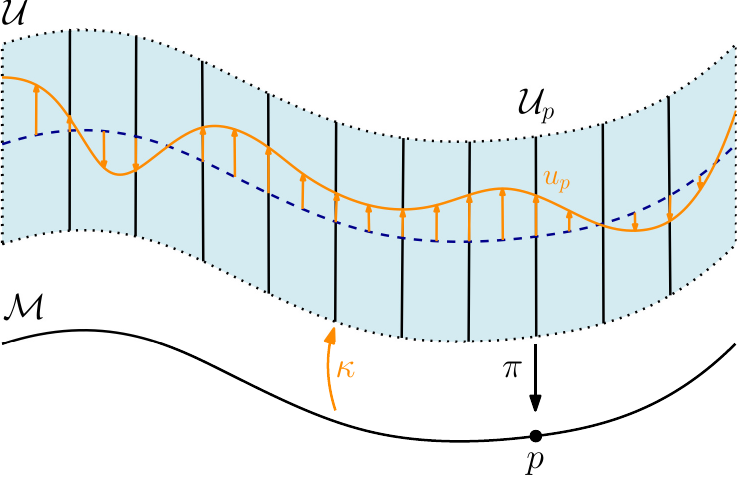}
    \caption{In our global geometric setting, a nonlinear control system is defined on a bundle $(\pi, \U, \M)$. Fibers $\mathcal U_p$ are spaces of control inputs which can be applied at $p$. A controller $\kappa$ is a map taking $p \in \M$ to $\kappa(p) \in \U_p$; above, the image of $\M$ under a controller $\kappa$ is drawn in orange.}
    \label{fig:vect-bundle-sys}
\end{figure}

\subsection{Safety from Control Barrier Functions}
With an appropriate class of nonlinear systems in place, we are now ready to build the theory of CBFs in the geometric setting. Consider a safe set $\C$ defined by\footnote{We consider $C^\infty$ objects for simple exposition; this is relaxable to \textit{differentiable as many times as necessary}, with $C^1$ and $C^2$ usually sufficing.} $h \in C^\infty(\M)$,
\begin{align}\label{eq:safeset_manifold}
    \C = \setdef{p \in \M}{ h(p) \geq 0},
\end{align}
Provided zero is a regular value of $h$, $\C \subseteq \M$ is a properly embedded, codimension-0 submanifold with boundary $\partial \C = h^{-1}(0)$, which itself is a properly embedded submanifold of $\M$ of codimension 1 \cite[Prop. 5.47]{lee2013smooth}. As in $\R^n$, we frame safety with respect to $\C$ as \textit{forward invariance} of $\C$.
\begin{defn}
    Let $X \in \X(\M)$. A set $\C \subseteq \M$ is \textit{forward invariant} for $X$ if the flow $\varphi_t(p_0)$ from any $p_0\in\C$ remains in the set $\C$ for all times $t \in I_{\geq 0}(p_0)$.
\end{defn}
\begin{theo}\label{theo:fwd-inv}
    Let $X$ be a locally Lipschitz vector field on $\M$, and $\C$ be as in \eqref{eq:safeset_manifold}. $\C$ is forward invariant for $X$ if either:
    \begin{enumerate}[(i)]
        \item zero is a regular value of $h$, and $dh_p X_p \geq 0, \; \forall p \in \partial \C$;
        \item there exists a locally Lipschitz $\alpha \in \K_\infty^e$ and an open set $\D \supset \C$ for which $dh_p X_p \geq - \alpha(h(p)), \; \forall p \in \D$.
    \end{enumerate}
\end{theo}
\begin{proof}
    We prove (i) using a corollary of Nagumo's theorem. We recall from \cite[Ex. 4.1.29]{abraham2012manifolds} that, for a vector field $Y \in \X(\R^n)$ and a map $g \in C^\infty(\R^n)$ with zero a regular value, $g^{-1}(\R_{\geq 0})$ is forward-invariant for $Y$ if $\tfrac{\partial g}{\partial \bx}|_{\bx} Y_{\bx} \geq 0, \; \forall \bx \in g^{-1}(0)$. Thus, if $\tfrac{\partial g}{\partial \bx}|_{\bx} Y_{\bx} \geq 0$ for all $\bx \in g^{-1}(0) \cap \V$, for $\V$ an open set, then for any $\bx_0 \in g^{-1}(0) \cap \V$, there is an $\epsilon > 0$ for which $\varphi_t(\bx_0) \in g^{-1}(\R_{\geq 0}) \cap \V$, for all $t \in [0, \epsilon)$.

    We now use this result to establish (i). Since any integral curve of $X$ starting in and leaving $\C$ must cross through $\partial \C$, it is sufficient to prove that for any $p_0 \in \partial \C$, there is an $\epsilon > 0$ for which $\varphi_t(p_0) \in \C, \; \forall t \in [0, \epsilon)$. Fix any $p_0\in \partial \C$; here, $h(p_0)=0$. Since zero is a regular value of $h$, there is a chart $(\mathcal V, \psi = (\bx^i))$ for $\M$ for which $\psi(p_0) = 0$ and $\C \cap \mathcal V = \setdef{(\bx^1, ..., \bx^n) \in \mathcal V}{\bx^n \geq 0}$. Let $\hat h$ and $\hat X$ be the local representatives of $h$ and $X$. The condition $dh_{p} X_{p} \geq 0, \; \forall p \in \partial \C$ implies $\frac{\partial \hat h}{\partial \bx}|_{\bx} \hat X_{\bx} \geq 0, \; \forall \bx \in \psi(\partial \C) \cap \psi(\V)$. We conclude there exists an $\epsilon > 0$ for which $\varphi_t(p_0) \in \C$ for all $t \in [0, \epsilon)$. This implies $\C$ is forward invariant for $X$.

    Now, we prove (ii). Let $\varphi_t^{\D}(p_0)$ be the flow of the restriction $X|_{\D}$ of $X$ to $\D$, and let $I^{\D}(p_0)$ be its maximal interval of existence. Since $h(\varphi_t^{\D}(p_0))$, viewed as a function of $t$, maps from $I^{\D}(p_0) \subseteq \R$ to $\R$, we may apply the comparison lemma \cite[Lemma 2.5]{khalil2002nonlinear}. The condition $dh_p f(\kappa(p)) \geq -\alpha(h(p)), \forall p \in \D$ implies the existence of a $\beta \in \KL$ for which $h(\varphi_t^{\D}(p_0)) \geq \beta(h(p_0), t), \forall t \in I_{\geq 0}^{\D}(p_0)$ and all $p \in \D$. Thus, $p_0 \in \C$ implies $\varphi_t^{\D}(p_0) \in \C, \forall t \in I^{\D}_{\geq 0}(p_0)$.
    
    Since $\varphi_t^{\D}(p_0)$ remains in $\C$ for all $t \in I_{\geq 0}^{\D}(p_0)$, it must be that $I_{\geq 0}^{\D}(p_0) = I_{\geq 0}(p_0)$. Uniqueness of flows \cite[Theorem 9.12]{lee2013smooth} then implies $\varphi_t^{\D}(p_0) = \varphi_t(p_0)$ for all $t \in I_{\geq 0}(p_0)$, from which we conclude forward invariance of $\C$ for $X$.
\end{proof}

Now, let $\Sigma = (\pi, \U, \M, \bPhi)$ be a nonlinear control system. For any controller $\kappa$ for $\Sigma$, the closed-loop system $\bPhi \circ \kappa$ is a vector field on $\M$. Thus, if we can design a controller $\kappa$ for which $\bPhi \circ \kappa$ satisfies either of the conditions of Theorem \ref{theo:fwd-inv}, then $\C$ will be an invariant set for the closed-loop system. This motivates the following, geometric definition of a CBF.
\begin{defn}\label{defn:cbf}
    A function $h \in C^\infty(\M)$ defining a safe set $\C$ as in~\eqref{eq:safeset_manifold} is a \textit{control barrier function} (CBF) for system \eqref{eq:nonlinear-manifold} if there exists a locally Lipschitz $\alpha \in \K_\infty^e$ for which
    \begin{align}\label{eqn:cbf-constr}
        \sup_{u_p \in \U_p} dh_p \bPhi(u_p) > -\alpha(h(p)), \; \forall p \in \C.
    \end{align}
    If additionally, \eqref{eqn:cbf-constr} holds on an open set $\D \supset \C$, $h$ is said to be a CBF for \eqref{eq:nonlinear-manifold} on $\D$.
\end{defn}
\begin{remark}
    The strict inequality \eqref{eqn:cbf-constr} follows the modern CBF literature. This implies zero is a regular value of $h$ and provides additional regularity in controller design.
\end{remark}
This definition ensures that at each point $p \in \D$, there is a $u_p$ that can satisfy condition (i) or (ii) of Theorem \ref{theo:fwd-inv}. However, as in Euclidean setting, without additional structure on the dynamics, it is difficult to tie these individual inputs together into a locally Lipschitz controller. Once again, the structure that enables safe, locally Lipschitz controller synthesis is provided by a control-affine system.

\subsection{CBF-based Controllers for Control-Affine Systems}
By specializing the bundle $(\pi, \U, \M)$ to a vector bundle and the map $\bPhi$ to one that is affine on each fiber $\U_p$, we arrive at the definition of a control-affine system. 
\begin{defn}\label{defn:aff-control-sys}
    A nonlinear control system $\Sigma = (\pi, \U, \M, \bPhi)$ is \textit{control-affine} if $(\pi, \U, \M)$ is a vector bundle of rank $m$ and for each $p \in \M$, the restriction $\bPhi|_{\U_p} : \U_p \to T_p \Mc$ of $\bPhi$ to $\U_p$ is an affine map. If additionally, $\U$ has a vector bundle metric and $\M$ a Riemannian metric,\footnote{When required, we will distinguish between these by $\braket{\cdot, \cdot}_{\U}$ and $\braket{\cdot, \cdot}_{\M}$.} $\Sigma$ is said to be \textit{metric}.
\end{defn}
\begin{example}\label{ex:single-integr}
    The \textit{single integrator} control system on $\M$, $(\pi_{T\M}, T\M, \M, \id_{T\M})$, is a control-affine system for which controllers are vector fields on $\M$.
\end{example}
The affine structure of $\bPhi$ implies the existence of a unique $\bf \in \X(\M)$ and a vector bundle map $\bG: \U \to T\Mc$ for which: 
\begin{align}\label{eq:contr-aff-geometric}
    \dot p = \bPhi(u_p) = \bf_p + \bG_p u_p,
\end{align}
on each fiber $\U_p$ of the bundle. Here, $\bf$ and $\bG$ generalize the familiar drift and control terms of \eqref{sys:ctrl-affine_euclidean}. 

We now develop the theory of safety filters for the class of metric control-affine systems. First, we explore the common technique of using the barrier condition to formulate a controller from a quadratic program (CBF-QP):
\begin{subequations}\label{eq:cbf-qp}
\begin{align}
        \kappa_\textup{QP}(p) = &\argmin_{u_p \in \U_p} \norm{u_p - \kappa_\des(p)}_{\U}^2 \label{eq:cbf-qp-cost}\\
        &\quad\textup{s.t.} \quad   dh_p\bf_p + dh_p \bG_p u_p  \geq -\alpha(h(p)).
        \label{eq:cbf-kappa}
    \end{align}
\end{subequations}
This controller minimally modifies a desired controller $\kappa_\des\in\Gamma(\U)$ in order to enforce safety. However, since it is not possible to \textit{directly} implement this problem in an optimization solver, and a numerical implementation for an abstract manifold would generally require fixing a local frame, we provide a global, analytical solution.
\begin{prop}\label{prop:safety-filter}
    Let $\Sigma$ be a metric control system with a given desired controller $\kappa_\des \in \Gamma(\U)$. If $h$ is a CBF for $\Sigma$ on $\D \supset \C$ with $\alpha \in \K_\infty^e$, then for each $p \in \D$, the unique optimal solution to the CBF-QP~\eqref{eq:cbf-qp}
    is given by: 
    $$
    \kappa_\textup{QP}(p) = \kappa_\des(p) + \lambda_\textup{QP}(a(p), b(p)) \bG_p^* \grad h_p,
    $$
    where $\bG_p^*$ is the adjoint\footnote{The \textit{adjoint} of $\bG_p$ with respect to the metrics on $\U$ and $\M$ is defined $\braket{\bG_p^* v_p, u_p}_{\U} = \braket{v_p, \bG_p u_p}_{\M} \; \forall v_p \in T_p \M, u_p \in \U_p$.} of $\bG_p$, $\grad h$ is the Riemannian gradient of $h$ on $\M$, and $\lambda_{\textup{QP}}$, $a$, and $b$ are defined:
    \begin{align}\label{eq:lambda_qp}
        &\lambda_\textup{QP}(a, b) = \begin{dcases}
            0 & b = 0\\
            \max\{0, -\tfrac{a}{b}\} & b\neq 0,
        \end{dcases} \; (a, b) \in \R^2,\\
        &a(p) =  \alpha(h(p)) + \braket{\grad h_p, \bf_p}_{\M} + \braket{\grad h_p, \bG_p \kappa_\des(p)}_{\M},\nonumber\\
        &b(p) = \norm{\bG_p^* \grad h_p}_{\U}^2. \nonumber
    \end{align} 
\end{prop}
\begin{proof}
    We note that for each $p \in \D$, the constraint~\eqref{eq:cbf-kappa} is equivalently written in terms of the Riemannian gradient as:
    \begin{align}
        \braket{\grad h_p, \bf_p}_{\M} + \braket{\grad h_p, \bG_p u_p}_{\M} \geq -\alpha(h(p)). \label{eq:gradient-constraint}
    \end{align}
    With this in mind, we fix $p \in \D$ and split into cases.
    \begin{enumerate}
        \item \underline{Case I}: $b(p)=0$. Here, $\bG_p^* \grad h_p = 0$, and the condition~\eqref{eqn:cbf-constr} from the definition of CBF reduces to $dh_p\bf_p > -\alpha(h(p))$. This implies the constraint~\eqref{eq:cbf-kappa} is satisfied for any $u_p \in \U_p$. Hence, the optimal solution is $\hat u_p = \kappa_\des(p)$, matching the proposed solution for $\kappa_\textup{QP}$.
        \item \underline{Case II}: $b(p)\neq 0$ and $a(p) \geq 0$. The latter implies that $\kappa_\des(p)$ itself satisfies \eqref{eq:cbf-kappa}. Again, the optimal solution is $\hat u_p = \kappa_\des(p)$, which matches the proposed solution since $\lambda_\textup{QP}$ evaluates to zero. 
        \item \underline{Case III}: $b(p)\neq 0$ and $a(p) < 0$. We begin by establishing a lower bound on \eqref{eq:cbf-qp-cost} over the feasible set. From \eqref{eq:gradient-constraint}, $u_p \in \U_p$ is feasible if and only if:
        \begin{align}
            \braket{\grad h_p, \bG_p(u_p - \kappa_\des(p))}_{\M} &\geq -a(p)\\
            \Longleftrightarrow \braket{\bG_p^* \grad h_p, u_p - \kappa_\des(p)}_{\U} &\geq -a(p). \label{eq:adjoint-inequality}
        \end{align}
        Since $-a(p)>0$, we may square both sides of the inequality \eqref{eq:adjoint-inequality} and apply Cauchy-Schwarz to conclude:
        \begin{align}
             &\norm{\bG_p^* \grad h_p}_{\U}^2 \norm{u_p - \kappa_\des(p)}_{\U}^2 \geq a(p)^2\label{eq:CS-Ineq}\\
             \implies &\norm{u_p - \kappa_\des(p)}_{\U}^2 \geq \frac{a(p)^2}{\norm{\bG_p^* \grad h_p}_{\U}^2},
        \end{align}
        which gives us a lower bound on the cost function \eqref{eq:cbf-qp-cost} that is both uniform in $u_p$ and valid over the feasible set at $p$. By substituting our proposed solution $\hat u_p$, 
        \begin{align}
           \hat u_p = \kappa_\des(p) - \frac{a(p)}{b(p)}\bG_p^* \grad h_p,
        \end{align}
        for $u_p$, we verify that $\hat u_p$ is both feasible and achieves the lower bound, so it must be optimal at $p$. 
    \end{enumerate}
    Uniqueness of the CBF-QP solution follows from strict convexity of the cost and convexity of the constraint.
\end{proof}
Proposition \ref{prop:safety-filter} gives a closed-form expression for the optimal safety filter that lets us implement safety-critical geometric controllers without an optimization solver. Further, its proof does not require any heavy machinery from optimization, relying solely on simple lower bound arguments; this suggests the potential for further generalization \cite{ames2025control}.
\begin{cor}\label{cor:loc-lip-safety-filter}
    Provided $\kappa_\des$ and $\alpha$ are locally Lipschitz, $\kappa_{\textup{QP}}$ is locally Lipschitz on $\D$ and renders $\C$ forward invariant for the closed-loop vector field $\bPhi \circ \kappa_{\textup{QP}}$ on $\D$.
\end{cor}
\begin{proof}
    First, we note that the function $\lambda_\textup{QP}$ \eqref{eq:lambda_qp} is locally Lipschitz on the set $\Pc = \setdef{(a, b) \in \R^2 }{a > 0 \vee b > 0}$. Since $h$ is a CBF on $\D$, $b(p) = 0$ implies $a(p) > 0$, and $a(p) \leq 0$ implies $b(p) > 0$. Hence, $(a(p), b(p)) \in \Pc$ for all $p \in \D$. Since $a$ and $b$ are locally Lipschitz and the vector bundle operations of addition and scalar multiplication are smooth, $\kappa_{\textup{QP}}$ is locally Lipschitz on $\D$. Thus, $\bPhi \circ \kappa_{\textup{QP}}$ is a locally Lipschitz vector field on $\D$. By Theorem \ref{theo:fwd-inv}, we conclude that $\C$ is forward invariant for $\bPhi \circ \kappa_{\textup{QP}}$.
\end{proof}

Although local Lipschitz continuity suffices for safety, smoothness is desirable in many contexts. For example, smoothness will be required later when we apply backstepping on manifolds to extend  our safety results to mechanical systems.
To this end, we study the geometric analogue of smooth safety filters, for which explicit formulas are also available~\cite{sontag1989universal,cohen2023characterizing}. In particular, the \textit{half-Sontag} universal formula is built on the following function:
\begin{equation}\label{eq:lambda_half_sontag}
    \lambda_\textup{HS}(a,b) = \begin{cases} 0 & b= 0, \\
        \frac{-a + \sqrt{a^2 + b^2}}{2b} & b\neq0,
    \end{cases}\quad (a, b) \in \R^2,
\end{equation}
which, despite its appearance, is an analytic function on the set $\Pc=\setdef{(a,b) \in \R^2}{a > 0 \vee b > 0}$ \cite{cohen2023characterizing}. This function is a smooth \textit{squareplus} overapproximation of $\lambda_\textup{QP}$ in \eqref{eq:lambda_qp}.

\begin{cor}\label{cor:smooth-safety-filter}
    Let $\Sigma$ be a control-affine system and $h$ a CBF for $\Sigma$ on an open neighborhood $\D \supset \C$ with a smooth $\alpha \in \K_\infty^e$. For any smooth $\kappa_\des \in \Gamma(\U)$, the map $\kappa_{\textup{HS}}: \D \to \U$,
    \begin{align}\label{eq:smooth-safety-filter}
        p \mapsto \kappa_{\textup{HS}}(p) = \kappa_\des(p) + \lambda_{\textup{HS}}(a(p), b(p))\bG_p^* \grad h_p ,
    \end{align}
    with $\lambda_{\text{HS}}$ the half-Sontag function \eqref{eq:lambda_half_sontag}, is smooth and satisfies $dh_p \bPhi(\kappa_{\textup{HS}}(p)) \geq -\alpha(h(p))$ for all $p \in \D$. Further, there is a globally-defined smooth controller $\tilde \kappa_{\textup{HS}} \in \Gamma(\U)$ for which $\tilde \kappa_\textup{HS}(p)=\kappa_\textup{HS}(p)$ on a neighborhood of $\C$ contained in $\D$.
\end{cor}
\begin{proof}
    As detailed above, $\lambda_{\textup{HS}}$ is analytic on $\Pc$; this implies smoothness of $\kappa_{\textup{HS}}$ on $\D$. By directly following \cite[Theorem 2]{cohen2023characterizing}, one verifies that $\kappa_{\textup{HS}}$ satisfies $dh_p \bPhi(\kappa_{\textup{HS}}(p)) \geq -\alpha(h(p))$ on $\D$. By \cite[Lemma 10.12]{lee2013smooth}, there exists a global, smooth extension $\tilde \kappa_{\textup{HS}} \in \Gamma(\U)$ of $\kappa_{\textup{HS}}$ which is supported in $\D$ and matches $\kappa_{\textup{HS}}$ on a smaller neighborhood of $\C$ in $\D$.
\end{proof}
We conclude that surprisingly little structure is required to retain the standard CBF controller synthesis techniques. Looking at the individual pieces of this result, this is \textit{not} at all surprising---here, we simply replace the familiar Euclidean vector spaces with abstract vector spaces (fibers of a vector bundle)---the bundle then smoothly stitches these vector spaces together to recover the standard CBF results.

%% file: cbf_synth.tex
\section{CBFs for Geometric Mechanical Systems}
In practical robotics problems, it is desirable to synthesize CBFs directly from configuration constraints; e.g. one might wish to turn obstacle locations into a CBF for a robotic system. In this spirit, we now study the synthesis of CBFs from configuration constraints for geometric mechanical systems.

\subsection{Mechanics on Riemannian Manifolds}
First, we review the background on geometric mechanical systems. We shall make use of the connection formulation of Lagrangian mechanics, which generalizes the familiar Euler-Lagrange equations in $\R^n$ to the Riemannian setting \cite{bullo2019geometric}.
\begin{defn}\label{defn:SMCS}
    A \textit{simple mechanical control system (SMCS)} is a tuple $\Sigma = (Q, g, V, \F)$, of an $n$-dimensional manifold $Q$ (the \textit{configuration manifold}), a Riemannian metric $g$ (the \textit{kinetic energy metric}), $V \in C^\infty(Q)$ (the \textit{potential function}), and a rank $m$ codistribution $\F$ (the \textit{input codistribution}).
\end{defn}
\begin{defn}[Control Force]
    A \textit{control force} for a SMCS $\Sigma$ is a map $F: TQ \to \F$ for which $\pi_{T^*Q} \circ F = \pi_{TQ}$.
\end{defn}

To write the equations of motion of a SMCS, we require the \textit{Levi-Civita connection} $\nabla$, a map uniquely determined by $g$ that lets us compute directional derivatives on manifolds. For $X, Y \in \X(Q)$, $\nabla_X Y$ is the directional derivative of $Y$ in the direction $X$, termed the \textit{covariant derivative} of $Y$ in the direction $X$. Fortunately, many of the familiar directional derivative properties translate to the covariant derivative \cite{lee2018introduction}.

For a curve $q: I \to Q$ in the configuration manifold, $\nabla_{\dot q} \dot q$ is the \textit{geometric acceleration} of $q$, valued in $TQ$. In terms of $\nabla$, the equations of motion of a SMCS are written as\footnote{$\sharp$ is the cotangent-tangent isomorphism $\braket{\omega_q^\sharp, v_q} = \braket{\omega_q; v_q} \forall v_q \in TQ$.}:
\begin{align}\label{eqn:eom}
    \nabla_{\dot q}\dot q = - \grad V_q + (F(\dot q))^\sharp.
\end{align}
Importantly, \eqref{eqn:eom} is a \textit{global} description of the dynamics, non-reliant on local coordinate charts. For a specified (locally Lipschitz) control force $F$, \eqref{eqn:eom} lifts to a (locally Lipschitz) vector field in $\X(TQ)$. As such, $\forall v_q \in TQ$, there is a unique maximal integral curve $\dot q : I(v_q) \to TQ$ with $\dot q(0) = v_q$.

The description \eqref{eqn:eom} of the equations of motion determines a control-affine system with state manifold $TQ$ \cite{bullo2019geometric}.\footnote{One may take the input space $\U_{v_q}$ to be the \textit{vertical lift} of $\F_q^\sharp$; as the lift is an isomorphism, we may identify vectors in $\U_{v_q}$ and $\F_q$ \cite[p. 68]{oliva2004geometric}. This let us unambiguously write $\bPhi$ as a function of $(v_q, F_q)$ or $F(v_q)$.} Instead of working with the induced control-affine system, however, we will find it more analytically convenient to work within the Levi-Civita connection formalism; this will avoid any complex calculations on $T(TQ)$. Although this will occasionally result in a slight abuse of definitions, the results are equivalent to those for the control-affine system on $TQ$.

To perform CBF synthesis for a SMCS, we must understand its \textit{actuation geometry}. This is encapsulated by two distributions, based on those introduced in \cite{welde2022role}.
\begin{defn}
    Let $\Sigma$ be a SMCS. The \textit{unactuated distribution} $\Uc Q$ of $\Sigma$ is the coannihilator of the input codistribution,
    \begin{align}
        \Uc_qQ = \{v_q \in T_q Q: \braket{F_q; v_q} = 0,\;  \forall F_q \in \F_q\},
    \end{align}
    while the \textit{actuated distribution} $\Ac Q$ is its $g$-orthogonal complement, $\Ac_qQ = \Uc_qQ^\perp.$
\end{defn}

By construction, at each configuration $q \in Q$ the tangent space $T_q Q$ splits as the direct sum $T_q Q = \Ac_q Q \oplus \Uc_q Q$. Thus, for every $v_q \in T_qQ$, there are unique $v_q^\Ac \in \Ac_qQ$ and $v_q^\Uc \in \Uc_qQ$ for which $v_q = v_q^\Ac + v_q^\Uc$. This splitting induces vector bundle maps $\pi_\Ac, \pi_\Uc: TQ \to TQ$, which project $v_q \mapsto v_q^\Ac$ and $v_q \mapsto v_q^\Uc$, respectively. Later, we will find it useful to understand the interaction between such maps and covariant differentiation, which we summarize as follows.
\begin{lem}\label{lem:tens-field}
    Let $(Q, g)$ be a Riemannian manifold and $X, Y \in \X(Q)$. For any vector bundle map $A: TQ \to TQ$,
    \begin{align}
        \nabla_X(AY) &= A \nabla_X Y + (\nabla_X A)(Y),
    \end{align}
    for $\nabla_X A$ the covariant derivative of $A$ as a tensor field.\footnote{A $(1, 1)$-tensor field on $Q$ is a smooth assignment of $q \in Q$ to a bilinear map from $T_q^*Q \times T_qQ  \to \R$, $q \mapsto \tilde A_q : T_q^* Q \times T_qQ \to \R$. Just like we can identify a matrix $B \in \R^{n\times n}$ with a bilinear form $\tilde B(x, y) = x^\top B y$, we can identify $A$ with a $(1, 1)$-tensor field $\tilde A_p(\omega_p, v_p) = \braket{\omega_p; Av_p}$. As the covariant derivative of tensor fields is well-defined, $\nabla_X A$ is well-defined.}
\end{lem}
\begin{proof}
    We calculate the covariant derivative of $A$ as a $(1, 1)$ tensor field. Applying \cite[Proposition 4.15 (a), (b)]{lee2018introduction} for any $X, Y \in \X(Q)$ and $\omega \in \Gamma(T^* Q)$ yields
    \begin{align}
        (\nabla_X A)(\omega, Y) = \braket{\omega; \nabla_X(AY)} - \braket{\omega; A \nabla_X Y}. 
    \end{align}
    As this holds for all $\omega \in \Gamma(T^*Q)$, we pass back through the tensor field-VB map identification to conclude $(\nabla_X A)(Y) = \nabla_X(AY) - A \nabla_X Y$. Rearranging, the result follows.
\end{proof}

\subsection{Backstepping Construction}
Our goal is to synthesize a CBF for a simple mechanical control system from a safety constraint on its configuration manifold~$Q$. Since these systems are second-order, we must employ high-order CBF techniques to translate such a constraint into a CBF. Here, we propose one such technique based on \textit{backstepping}, in which we lift a configuration safety specification on $Q$ to a CBF on $TQ$ by penalizing error to a safe configuration-level vector field on $Q$. This approach avoids the challenges associated with the ``standard" HOCBF, which may not produce valid CBFs even for simple systems and constraints~\cite{cohen2024safety-review}. Further, our backstepping technique will be integrated with the SMCS framework in a manner that avoids computations on the second-order tangent bundle $T(TQ)$. This CBF synthesis process is outlined as follows:
\begin{enumerate}
    \item Identify a configuration safety specification $h_0$ on $Q$.
    \item Find a smooth, safe velocity vector field $\kappa$ for $h_0$ on $Q$.
    \item Use $h_0, \kappa$ in backstepping to produce a CBF for $\Sigma$.
\end{enumerate}
We begin filling in the technical details of this process by precisely defining a configuration safety specification.
\begin{defn}\label{defn:CSS}
    A \textit{configuration safety specification} on a manifold $Q$ is a function $h_0 \in C^\infty(Q)$ for which zero is a regular value. We denote its zero superlevel set $h_0^{-1}(\R_{\geq 0})$ by $\C_0$. 
\end{defn}

Our first task is to generate a smooth, safe velocity vector field on $Q$ for the configuration safety specification $h_0$. To accomplish this, we will construct a smooth vector field $\kappa \in \X(Q)$ for which $d(h_0)_q \kappa_q > -\alpha(h_0(q))$, for all $q$ in an open set $\D_0 \supset \C_0$ and some $\alpha \in \mathcal{K}_\infty^e$. Recalling from Example \ref{ex:single-integr} the \textit{single integrator} control system on a manifold, we will show that designing $\kappa$ is reducible to designing a smooth safety filter for the single integrator on $Q$.

\begin{lem}\label{lem:config-safety}
    Let $h_0$ be as in Definition \ref{defn:CSS} and $\kappa_{\textup{HS}}^Q \in \X(Q)$ be any smooth safety filter for the single integrator on $Q$, as in Corollary \ref{cor:smooth-safety-filter}. There is an open set $\D_0 \supset \C_0$ on which:
    \begin{align}
        \kappa_q = \kappa_{\textup{HS}}^Q(q) + \delta \grad h_0|_q
    \end{align}
    satisfies $d(h_0)_q \kappa_q > -\alpha(h(q)) \; \forall q \in \D_0$, for any $\delta > 0$.
\end{lem}
\begin{proof}
   To show $\kappa_{\textup{HS}}^Q$ is well-defined, we will show $h_0$ is a CBF for the single integrator on $Q$. Fix any smooth $\alpha \in \K_\infty^e$, and let $\Vc \supset \C_0$ be an open set for which $\Vc \setminus \C_0$ contains no critical points ($\V_c$ is guaranteed to exist since $0$ is a regular value of $h_0$). For any $q \in \Vc$, either $h_0(q)>0$ or $d(h_0)_q$ is nonzero, implying $\sup_{v_q \in T_q Q} d(h_0)_q v_q > -\alpha(h_0(q))$. Thus, $h_0$ is a CBF on $\V$ for the single integrator. Proposition \ref{cor:smooth-safety-filter} yields a smooth $\kappa_{\textup{HS}}^Q \in \X(Q)$ satisfying a non-strict barrier inequality on an open set $\D_0 \supset \C_0$; for any $\delta > 0$, adding $\delta \grad h_0$ to $\kappa_{\textup{HS}}^Q$ makes the inequality strict on~$\D_0$.
\end{proof}

Thus, given \textit{any} configuration safety specification, we can find a smooth, globally-defined safe vector field on $Q$ that satisfies a strict barrier condition on a neighborhood of $\C_0$. In the case of a compact safe configuration set $\C_0$, the Riemannian gradient $\grad h_0$ also provides a safe velocity vector field on a neighborhood of $\C_0$. This situation occurs, for instance, if $h_0$ is constructed from a \textit{navigation function} in the sense of \cite{koditschek1990robot} (e.g., $h_0 \defeq 1 - \textup{Nav. Func.}$).

Now, we discuss the construction of the CBF for a mechanical system. Inspired by the CBF backstepping method of \cite{cohen2024safety-review}, we consider the CBF candidate:
\begin{align}\label{eqn:cand-cbf}
    h(v_q) = h_0(q) - \frac{\varepsilon}{2}\norm{(v_q - \kappa_q)^\Ac}^2, \; v_q \in TQ,
\end{align}
where $\varepsilon > 0$ is a fixed design parameter and $\norm{\cdot}$ is induced by the kinetic energy metric. We interpret the norm term in \eqref{eqn:cand-cbf} as penalizing the projection of error to the safe velocity $\kappa_q$ onto the space of directions in which we can actuate. The scalar $\varepsilon$ tunes how close $h$ is to the constraint $h_0$ on $Q$; smaller $\varepsilon$ yields a function that better approximates the safe configuration set, potentially at the cost of higher actuation.

Now, we begin the process of proving \eqref{eqn:cand-cbf} is a CBF. Here, we will find it much easier to work with \eqref{eqn:eom} instead of the induced control-affine dynamics $\bPhi$ on $TQ$. To emphasize this choice, we will write $\dot h(v_q, F_q)$ in place of $dh_{v_q} \bPhi(v_q, F_q)$; e.g. \eqref{eqn:cbf-constr} can be written $\sup_{F_q \in \F_q} \dot h(v_q, F_q) > -\alpha(h(v_q))$. As $\dot h(v_q, F_q)$ and $dh_{v_q} \bPhi(v_q, F_q)$ are equivalent, our CBF theory immediately translates; we make the notational distinction simply to emphasize our use of the Levi-Civita framework.
\begin{lem}\label{lem:cbf-lie-deriv}
    Let $\Sigma$ be a SMCS and $\kappa \in \X(Q)$. Consider the CBF candidate $h$ (\ref{eqn:cand-cbf}). For any $v_q \in T_qQ$ and $F_q \in \F_q$,
    \begin{align}\label{eqn:hDot}
         &\dot h(v_q, F_q) = d(h_0)_q v_q
        +  \varepsilon \langle e_q^\Ac,  \\
        &\nabla_{v_q}\kappa_{q} + \grad V_{q} - (\nabla_{v_q} \pi_A)(e_q) \rangle
        - \varepsilon \braket{F_q; e_q^\Ac}, \nonumber
    \end{align}
    where $e_q = v_q - \kappa_q$ and $\nabla_{v_q} \pi_A$ is the covariant derivative of $\pi_A$ as a $(1, 1)$-tensor field.
\end{lem}
\begin{proof}
    Let $\dot q: I(v_q) \to TQ$ be the unique maximal solution to \eqref{eqn:eom} under a force controller $F$ with $F(v_q) = F_q$, satisfying $\dot q(0) = v_q$. Suppressing $(t)$ for convenience, we successively use \cite[Prop. 5.5]{lee2018introduction} and Lemma \ref{lem:tens-field} to compute,
    \begin{align}
        \tfrac{d}{dt}\tfrac{1}{2}\norm{e^\Ac}^2 &= \braket{e^\Ac, \nabla_{\dot q} (e^\Ac)}\\
        &= \braket{e^\Ac, (\nabla_{\dot q} e)^\Ac + (\nabla_{\dot q} \pi_A)(e)}. \label{eq:error-dot}
    \end{align}
    Applying linearity ($\nabla_{\dot q} e = \nabla_{\dot q} \dot q - \nabla_{\dot q}\kappa_q$) and \eqref{eqn:eom} to \eqref{eq:error-dot},
    \begin{align}
        &=\braket{e^\Ac, (\nabla_{\dot q} \pi_A)e - \nabla_{\dot q} \kappa_q}
        + \braket{e^\Ac, \nabla_{\dot q} \dot q}\\
        &= \braket{e^\Ac, (\nabla_{\dot q} \pi_A)e - \nabla_{\dot q} \kappa_{q} - \grad V_q}
        + \braket{F(\dot q); e^\Ac}. \label{eq:error-dot-final}
    \end{align}
    Since $\dot h(v_q, F_q) = \frac{d}{dt}|_{t = 0} h(\dot q)$, \eqref{eqn:hDot} follows by substituting \eqref{eq:error-dot-final} into $\frac{d}{dt}|_{t = 0}h(\dot q) = (d(h_0)_q \dot q - \frac{\varepsilon}{2}\frac{d}{dt}\norm{e^\Ac}^2)|_{t = 0}$.
\end{proof}

\begin{theo}\label{theo:cbf-underact}
    Consider a SMCS $\Sigma = (Q, g, V, \F)$. Let $h_0$ and $\kappa \in \X(Q)$ be as in Lemma \ref{lem:config-safety}. If for all $q \in \D_0$, $\Uc Q$ satisfies $\U_qQ \subseteq \ker d(h_0)_q$, then for all $\varepsilon > 0$, \eqref{eqn:cand-cbf} is a CBF for $\Sigma$ on $T \D_0 \subseteq TQ$, with the same $\alpha \in \K_\infty^e$ as in Lemma \ref{lem:config-safety}.
\end{theo}
\begin{proof}
    We will prove that $\sup_{F_q \in \F_q}\dot h(v_q, F_q) > -\alpha(h(v_q))$ for all $v_q \in T\D_0$, where $\alpha$ is the same class $\K_\infty^e$ function that certifies the safety of $h_0$. Fix $q \in \D_0$ and $v_q \in T_q\D_0$.

    First, consider the case where $(v_q - \kappa_q)^\Ac = 0$. Fix any $F_q \in \F_q$. As $v_q$ admits a unique decomposition $v_q = v_q^\Ac + v_q^{\Uc}$, it must satisfy $v_q = \kappa_q^\Ac + v_q^\Uc$. Together with Lemma \ref{lem:cbf-lie-deriv} and the assumption of $\Uc_qQ \subseteq \ker d(h_0)_q$, this implies:
    \begin{align}
        \dot h(v_q, F_q) = d(h_0)_q v_q = d(h_0)_q \kappa_q^\Ac = d(h_0)_q \kappa_q. \label{eq:h-case-1}
    \end{align}
    By design, $\kappa$ satisfies $d(h_0)_q \kappa_q > -\alpha(h_0(q))$. Further, since $(v_q - \kappa_q)^\Ac = 0$, we must have $h_0(q) = h(v_q)$. We conclude that $\dot h(v_q, F_q) > -\alpha(h(v_q))$, completing the first case.

    Now, suppose $(v_q - \kappa_q)^{\Ac} \neq 0$. Here, there exists an $F_q \in \F_q$ for which $\braket{F_q; (v_q - \kappa_q)^\Ac} \neq 0$. Since Lemma \ref{lem:cbf-lie-deriv} establishes that $\dot h(v_q, F_q)$ is affine in $F_q$ with linear term $\braket{F_q; (v_q - \kappa_q)^\Ac}$, we must have $\sup_{F_q \in \F_q}\dot h(v_q, F_q) = +\infty$. We conclude that $\sup_{F_q \in \F_q} \dot h(v_q, F_q) > -\alpha(h(v_q))$ for all $v_q \in T\D_0$, and that $h$ is a CBF for $\Sigma$ on $T\D_0$.
\end{proof}

Theorem \ref{theo:cbf-underact} states that, provided the unactuated directions are not required for configuration safety, $h$ in \eqref{eqn:cand-cbf} is a CBF. Moreover, the derived expression for  $\dot h$ avoids computations in $T(TQ)$ and provides a tractable condition for enforcing safety. Note that in the fully actuated case, $\U_q Q \subseteq \ker d(h_0)_q$ is automatically satisfied for all $q \in Q$, and $\pi_{\Ac} = \id_{TQ}$.

The following proposition formalizes the relationship between \eqref{eqn:cand-cbf} and the satisfaction of configuration safety.

\begin{prop}
    Consider the setting of Theorem \ref{theo:cbf-underact}. Let $F$ be any locally Lipschitz control force for which $\dot h(v_q, F(v_q)) \geq -\alpha(h(v_q)), \; \forall v_q \in T\D_0$. For any safe initial velocity $\dot q(0) = \dot q_0, \; h(\dot q_0) \geq 0$, the trajectory $\dot q: I(\dot q_0) \to TQ$ of the closed-loop system satisfies $q(t) \in \C_0 \forall t \in I_{\geq 0}(\dot q_0)$.
\end{prop}
\begin{proof}
    By Theorem \ref{theo:fwd-inv}, $h(\dot q_0) \geq 0$ implies $h(\dot q(t)) \geq 0$ for all $t \in I_{\geq 0}(\dot q_0)$. Since $h_0(q) \geq h(v_q)$ for any $v_q \in TQ$, we conclude that $h_0(q(t)) \geq h(\dot q(t)) \geq 0 \; \forall t \in I_{\geq 0}(\dot q_0)$.
\end{proof}

Since $h$ is a CBF and $\dot h(v_q,F(v_q))\geq -\alpha(h(v_q))$ enforces safety, any of the geometric CBF control design methods introduced in Section III can be used to construct a safe control force~$F$. One particular choice is the (appropriate closed-form) CBF-QP controller that enforces the $\dot h$ constraint with the dual norm on $T^*Q$ used for the norm on $\F$ in the cost.

%% file: example.tex
\section{Application to the Satellite}
As an application of the above, we design a backstepping CBF for an underactuated satellite, which we  model as a rotating rigid body on $SO(3)$. We assume that the rigid body has inertia tensor $J = \diag(J_1, J_2, J_3)$ with $J_1 = J_2$, and is actuated about the $e_1$ and $e_2$ body-fixed axes. 

To remain safe, the satellite must keep its heat shield, which is orthogonal to the body-fixed $e_3$-axis, within a safe angle $\theta_{\textup{safe}}$ of the spatial $e_3$-axis. We encode this in the configuration safety specification $h_0(R) = e_3^\top R e_3 - \cos(\theta_{\text{safe}})$; the projection of the safe configuration set to $\sph^2$ is visualized in Figure \ref{fig:sphere-cap}. For the safe velocity vector field, we take the controller of Lemma \ref{lem:config-safety}; note that since $\C_0$ is compact, we can alternatively use the gradient of $h_0$ as the safe velocity vector field. As the underactuation condition of Theorem \ref{theo:cbf-underact} is satisfied, we take \eqref{eqn:cand-cbf} as our CBF, with the norm induced by the kinetic energy metric on $SO(3)$.

To ensure safety, we design a CBF-QP controller that filters the inputs from an unsafe nominal geometric PD controller. The nominal controller, obtained via symmetry reduction of the satellite to a fully-actuated system on $\sph^2$, tracks an unsafe reference trajectory. Figure~\ref{fig:sphere-cap} demonstrates that, while the nominal trajectory violates the configuration constraint, the CBF-QP keeps the satellite safe for all time.  

\begin{figure*}[ht]
    \centering
    \includegraphics[width=0.36\linewidth]{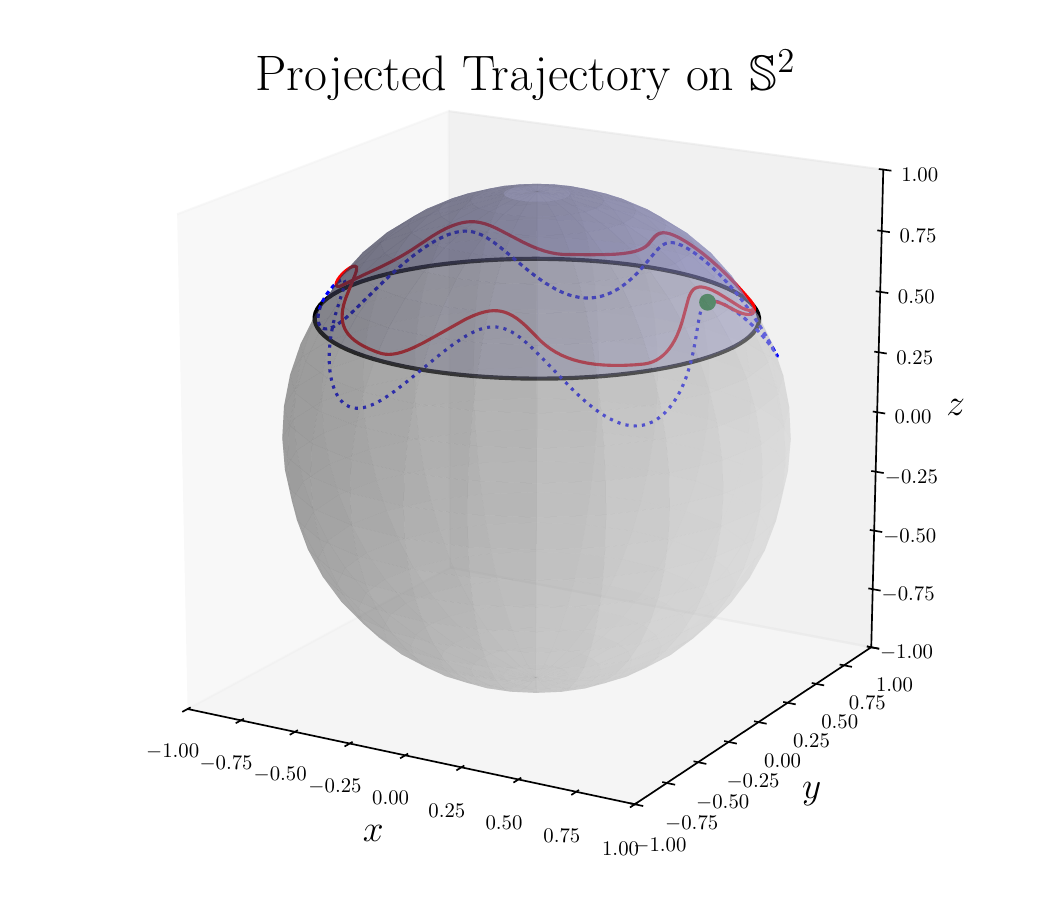}
    \includegraphics[width=0.3\linewidth]{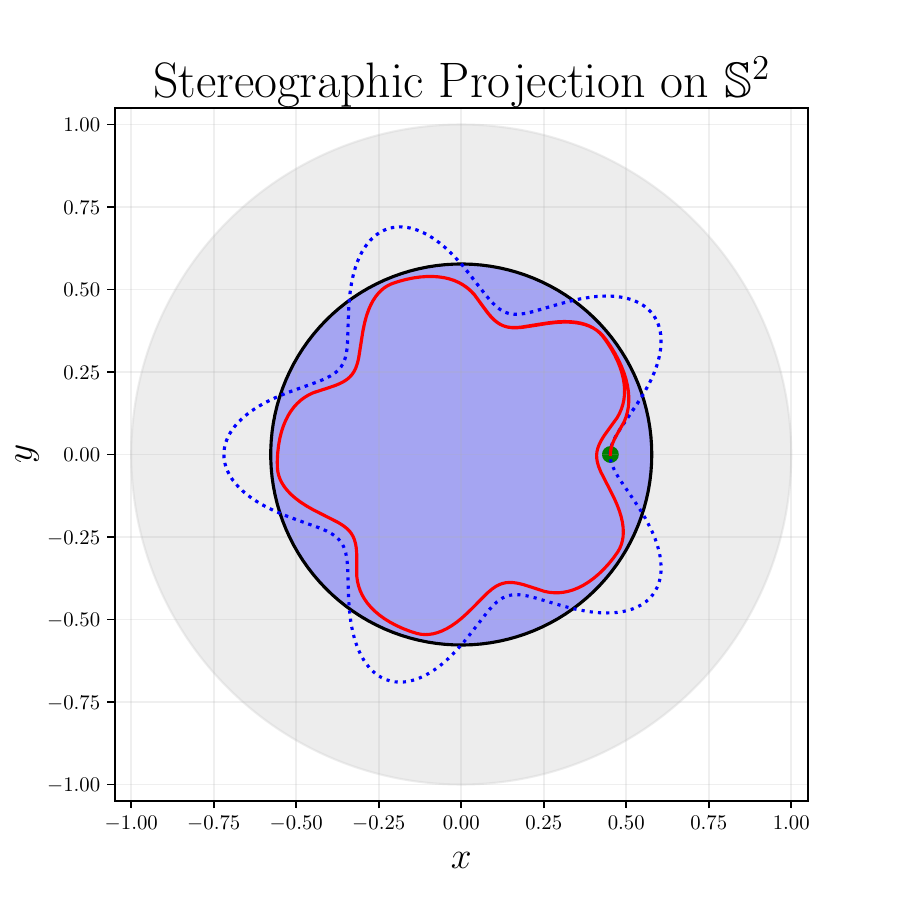}
    \includegraphics[width=0.29\linewidth]{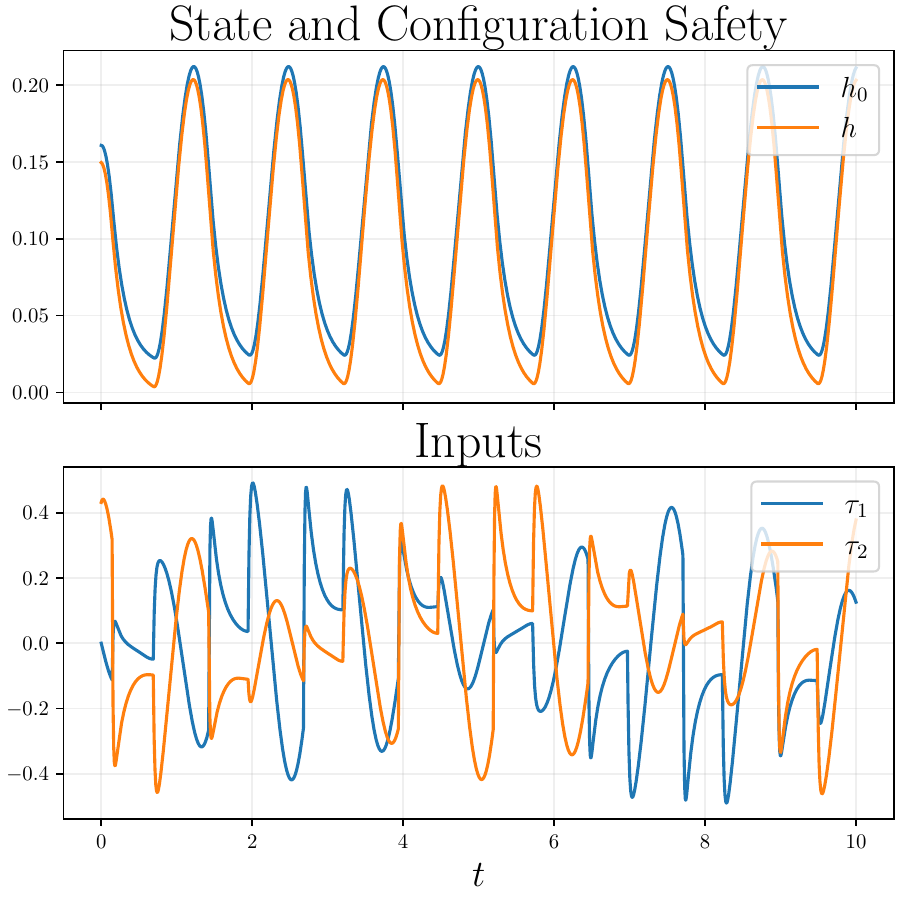}\\
    \includegraphics[width=0.9\linewidth]{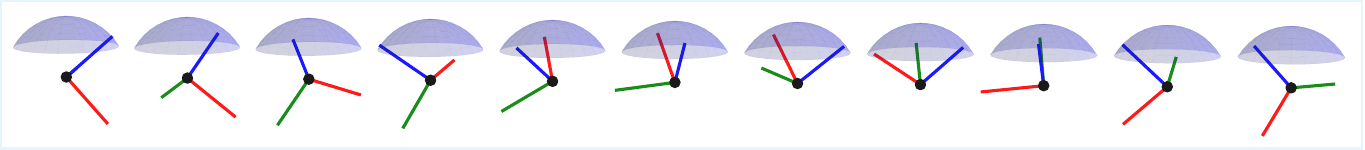}
    \caption{The safe set on $SO(3)$ is visualized by projecting to the sphere via $R \in SO(3) \mapsto Re_3 \in \sph^2$. A geometric tracking controller \cite{bullo2019geometric} for an unsafe trajectory (plotted as a dotted blue line), is filtered by a CBF-QP controller. The resulting trajectory, plotted in red, maintains a safe configuration for all time. The evolution of $R \in SO(3)$ is visualized as a sequence of frames, with the red, green, and blue axes representing $Re_1, Re_2, Re_3$, respectively.}
    \label{fig:sphere-cap}
\end{figure*}

%% file: conclusion.tex
\section{Conclusion}
In this work, we established a set of core theoretical results for CBFs on smooth manifolds. We began by formulating CBFs for systems on bundles, and then developed optimal safety filters and their smooth analogues for control-affine systems on vector bundles. In doing this, we demonstrated that the familiar Euclidean CBF results can be translated to the abstract geometric setting. Following this, we detailed a backstepping-based CBF synthesis method for mechanical systems, generalizing existing mechanical CBF techniques to a global, geometric setting. We then applied this technique to design a safety-critical controller for an underactuated satellite, demonstrating the efficacy of the backstepping construction in enforcing configuration-level safety.

Future work includes extending CBF synthesis to broader underactuation conditions, studying the robustness of the geometric techniques under uncertainty, and exploring the role of Lie group symmetries in facilitating CBF synthesis.